\documentclass[a4paper,UKenglish,cleveref, autoref, thm-restate]{lipics-v2021}

\hideLIPIcs  

\title{Efficient Generation of Expected-Degree Graphs via Edge-Arrivals} 

\author{Gianlorenzo D'Angelo}{Gran Sasso Science Institute, Italy}{gianlorenzo.dangelo@gssi.it}{https://orcid.org/0000-0003-0377-7037}{}

\author{Riccardo Michielan\footnote{Corresponding author. E-mail: riccardo.michielan@gssi.it}}{Gran Sasso Science Institute, Italy}{riccardo.michielan@gssi.it}{https://orcid.org/0000-0003-4642-8507}{}

\authorrunning{G. D'Angelo et al.} 

\Copyright{Gianlorenzo D'Angelo and Riccardo Michielan} 

\ccsdesc[100]{Theory of computation~Graph algorithms analysis}

\keywords{network models, graph generation, linear-time algorithm, alias sampling} 

\category{} 

\relatedversion{} 

\nolinenumbers 

\newcommand{\Prob}{\mathbb{P}}
\newcommand{\E}{\mathbb{E}}

\newcommand{\Poi}{\mathrm{Poisson}}
\usepackage{xspace}
\newcommand{\ER}{Erd\H{o}s--R\'enyi\xspace}

\usepackage{algorithm}
\usepackage{algpseudocode}
\usepackage{subcaption}

\ArticleNo{0}

\begin{document}

\maketitle

\begin{abstract}
We study the efficient generation of random graphs with a prescribed expected degree sequence, focusing on rank-1 inhomogeneous models in which vertices are assigned weights and edges are drawn independently with probabilities proportional to the product of endpoint weights. We adopt a temporal viewpoint, adding edges to the graph one at a time up to a fixed time horizon, and allowing for self-loops or duplicate edges in the first stage. Then, the simple projection of the resulting multigraph recovers exactly the simple Norros--Reittu random graph, whose expected degrees match the prescribed targets under mild conditions. Building on this representation, we develop an exact generator based on \textit{edge-arrivals} for expected-degree random graphs with running time $O(n+m)$, where $m$ is the number of generated edges, and hence proportional to the output size. This removes the typical vertex sorting used by widely-used fast generator algorithms based on \textit{edge-skipping} for rank-1 expected-degree models, which leads to a total running time of $O(n \log n + m)$. In addition, our algorithm is simpler than those in the literature, easy to implement, and very flexible, thus opening up to extensions to directed and temporal random graphs, generalization to higher-order structures, and improvements through parallelization.
\end{abstract}

\section{Introduction}
\label{sec:intro}

Generating large random graphs with controlled degree structure is a core task in network science, algorithm engineering, and simulation, where one often needs fast surrogates of observed networks for benchmarking \cite{drobyshevskiy2019random, fabius2026, schwartz2024}, uncertainty quantification \cite{zweig2016,casiraghi2021configuration}, dynamical experiments \cite{barrat2008,bartlett2012epidemic}, and, more broadly, as a computational laboratory for probing conjectures and searching for counterexamples \cite{roucairol2022refutation}. Several lines of work have therefore focused on scalable graph generation \cite{penschuck2022recent}, from linear-time algorithms for classical random graph models \cite{batagelj2005efficient}, to fast samplers for more realistic network ensembles \cite{miller2011efficient,meyer2022,bringmann2019geometric,blasius2019efficiently,kaminski2021,li2025fast}. 
In this paper, we focus on \emph{expected-degree} models: given a sequence of nonnegative weights $x=(x_i)_{i=1}^n$, encoding node-level propensities inferred from data, the goal is to sample a random graph such that $\E[\deg(i)] \approx x_i$ for all vertices $i \in [n]$.

A widely used and analytically tractable model family that achieves this objective is the class of \emph{rank-1} (or \emph{factorized}) inhomogeneous random graphs, in which edges are drawn independently, and the connection probability depends on the endpoints only through a product of their weights. Standard examples include the Chung--Lu model \cite{chung2002connected}, the generalized random graph \cite{britton2006generating,Hofstad_2016}, and the Poissonian Norros--Reittu construction \cite{norros2006conditionally}; under standard regularity assumptions, these formulations are asymptotically equivalent (see Sections 6.7-6.8 in \cite{Hofstad_2016}), and they lead to the desired expected degree sequence. A related hard-constrained alternative is the configuration model \cite{bollobas1980probabilistic,bender1978asymptotic}, which fixes the degree sequence exactly. In many applications, however, the weights $x$ are more naturally interpreted as intrinsic propensities or fitness parameters. Thus, from a modeling perspective, matching degrees only in expectation may be more robust to the idiosyncrasies of a single observed network realization.

A naive implementation of rank-1 models inspects all $\binom{n}{2}$ pairs of vertices, resulting in quadratic time even when the edge count is much smaller than quadratic. A prominent generative scheme for Chung-Lu graph sampling is an adaptation of the fast \ER generation algorithm \cite{batagelj2005efficient}, due to Miller and Hagberg \cite{miller2011efficient}: after sorting vertices by their weights, it avoids examining most non-edges via geometric jumps (\emph{edge-skipping}), and it achieves an expected running time $O(n+m)$ for the sampling phase. The general principle of skipping long runs of absent edges has also been recently applied to sample generalized random graphs~\cite{li2025fast}. Although the Miller--Hagberg algorithm works in expected time $O(n+m)$, it still requires the weights to be sorted in decreasing order. For arbitrary weights, this preprocessing is comparison-based and therefore requires $\Theta(n \log n)$ time in the worst case; only under special assumptions on the keys, such as bounded integer weights or fixed-length encodings, linear-time procedures such as counting or radix sort can be used instead. 
Alam et al.~\cite{alam2016efficient} and Moreno et al.~\cite{moreno2018scalable} improved the running time of Miller--Hagberg algorithm, but only for the case in which the nodes can be partitioned into sets with identical weights, thus their methods are not suitable for generic target expected degree sequences. Consequently, the overall complexity remains $O(n \log n + m)$, leaving open the problem of achieving truly linear-in-output generation time without any sorting step. Such a problem is especially relevant in the sparse regime $m = \Theta(n)$, which is the natural scale for many real-world networks \cite{newman2018networks, ghavasieh2024diversity, broido2019scale}.

Here we introduce a novel generative paradigm that addresses the above challenge by interpreting an expected-degree random graph as a \emph{snapshot} of an underlying temporal graph process. Specifically, each potential edge $\{i,j\}$ is endowed with a Poisson variable, whose rate is proportional to $x_i x_j$, determining its connectivity within a finite time horizon. Then, we construct the graph following an \textit{event-driven} formulation: the edge $\{i,j\}$ is included in the graph if and only if the edge event occurred at least once, that is, if the corresponding Poisson random variable is not null. In particular, we show that this generative algorithm is not an approximation of the Poissonized Norros--Reittu construction, but an alternative exact generative procedure, thus providing a flexible method for sampling expected-degree random graphs. In contrast to edge-skipping approaches, our edge-arrival method does not require any preliminary sorting of the weights: after $O(n)$ preprocessing, it samples a graph in $O(n+m)$ expected time. At a high level, the procedure repeatedly samples $m$ endpoint pairs from a distribution induced by $x$, implemented via the alias method. It then handles collisions, namely self-loops and repeated edges, through lightweight corrections, depending on whether the target object is the multigraph process itself or its projection to a simple graph. 

More broadly, our temporal graph construction paradigm via edge-arrivals is particularly simple to understand, easy to implement, and quite flexible in practice. Beyond simple undirected graphs, it naturally accommodates directed and temporal variants, and it suggests extensions to higher-order structures such as hypergraphs. Moreover, since new edges can be sampled independently across time windows or event batches, the approach also naturally lends itself to parallel implementation.

\paragraph*{Contributions}
\begin{itemize}
    \item We introduce a temporal multigraph construction and show that it coincides with the Poissonian Norros--Reittu model, yielding an exact generator for rank-1 expected-degree random graphs.
    \item We design an algorithm that, given the weight sequence as input, samples a graph in $O(n+m)$ expected time after $O(n)$ preprocessing, without any sorting step.
    \item We discuss how the same paradigm extends naturally to directed and temporal models, higher-order structures, and parallel implementations.
\end{itemize}

\section{Warm-up: \ER via temporal edge arrivals}
\label{sec:warmup}

We begin with an example that illustrates the temporal viewpoint in the simplest possible setting. Let $V = \{1,\dots,n\}$ and let $p \in (0,1)$. The \ER model is the random graph $G_{n,p}=(V,E)$, where each edge appears independently with probability $p$. The model construction below is based on edge activations, assuming the graph achieves a total amount of edges $m$, as in the original model \cite{erdos1960evolution,bollobas2001random}. While not identical, our temporal interpretation is closely related in spirit to standard dynamical \ER, as well as to recent work on random temporal graphs \cite{roberts2018exceptional,becker2023giant,broutin2024increasing}.\\

Rather than working directly with a simple graph, it is convenient to first define an underlying multigraph. Fix $T>0$, and let $\widetilde{G}_n(T)$ denote a multigraph on vertex set $V$ generated by a Poisson stream of edge activations of rate 1 up to time $T$. More precisely, let 
\begin{equation}
    \widetilde{m}\sim \Poi(T)
\end{equation} 
be the total number of activations observed up to time $T$. Each activation independently selects two endpoints $i,j\in V$ uniformly at random, with replacement, and inserts the corresponding edge $\{i,j\}$. Thus, at this stage, self-loops and repeated edges are allowed.

The associated simple graph $G_n(T) = (V, E(T))$ is obtained by \emph{erasing} the multigraph structure, namely, by removing self-loops and merging repeated edges. In particular, for each unordered pair $\{i,j\}$ with $i\neq j$, the number of activations of edge $\{i,j\}$ is Poisson with mean $2 T/n^2$, and these counts are independent across distinct pairs by Poisson thinning. Therefore,
\begin{equation}
    \Prob\left(\{i,j\}\in E(T)\right)=\Prob\left(\Poi\left(2T/n^2\right) \geq 1 \right) = 1-e^{-2T/n^2}.
\end{equation}
In particular, choosing \begin{equation}
    T = -\log(1-p)n^2/2,
\end{equation} 
one obtains $\Prob\left(\{i,j\}\in E(T)\right)=p$, independently over all pairs $\{i,j\}$, hence showing that the temporal formulation is equivalent to the original \ER random graph:
\begin{equation}
    G_n(T)\sim G_{n,p}.
\end{equation}

\subsection{Event-driven algorithm}

The multigraph construction above naturally suggests a generative algorithm. Conditional on the realized activation budget $\widetilde{m}$, the underlying temporal multigraph is generated by inserting $\widetilde{m}$ edges (events) one at a time, each obtained by sampling two endpoints independently and uniformly at random from $V$. The final simple graph is then obtained by removing self-loops and merging repeated edges.

Before giving a more formal description of the algorithm, we stress that $\widetilde{m}$ is not exactly distributed as the number of edges in $G_{n,p}$, which instead satisfies
\begin{equation}
    m := |E(G_{n,p})| \sim \text{Bin}\left(\binom{n}{2},p\right).
\end{equation}
In particular, when $p$ is close to $1$, the quantity $-\log(1-p)$ diverges, indicating that the activation budget required becomes extremely large. However, in more realistic settings when $p\ll 1$, the first-order expansion $-\log(1-p)=p+O(p^2)$
yields
\[
\E[\widetilde{m}] = \frac{n^2}{2}p + O(n^2p^2) = \E[m](1+o(1)),
\]
So the activation budget matches the edge count of $G_{n,p}$ asymptotically.

The resulting sampling strategy is summarized in Algorithm~\ref{alg:er-simple}. Rather than inspecting all $\binom{n}{2}$ candidate edges, the multigraph procedure reduces to the random choice of $\widetilde{m}$ pairs of vertices from $V$, by storing each non-loop edge in canonical order $(\min\{u,v\},\max\{u,v\})$ via a hash set. This provides a simple prototype of the general paradigm developed later for rank-1 models.

\begin{algorithm}
\caption{Event-driven algorithm for the simple \ER graph}
\label{alg:er-simple}
\begin{algorithmic}[1]
\Require number of vertices $n$, edge probability $p \in (0,1)$
\Ensure simple graph $G_{n,p}$ on vertex set $V=[n]$

\State Compute $T \gets -\log(1-p)\, n^2/2$
\State Sample $K \sim \Poi(T)$
\State Initialize an empty set of edges $E$

\For{$r=1$ to $K$}
    \State Sample $u$ uniformly from $V$
    \State Sample $v$ uniformly from $V$
    \If{$u \neq v$}
        \State Insert $(\min\{u,v\},\max\{u,v\})$ into $E$
    \EndIf
\EndFor

\State \Return $(V,E)$
\end{algorithmic}
\end{algorithm}
Sampling the activation budget $\widetilde{m}$ takes constant time, and each activation can be generated in $O(1)$ time by sampling two endpoints independently and uniformly from $V$. In the present homogeneous setting, this is immediate, whereas in the heterogeneous case considered later, the same role will be played by an alias-based sampler after linear preprocessing. The subsequent erasure step, which removes self-loops and merges repeated edges, can also be implemented in linear time with respect to the size of the generated multigraph using standard hashing-based data structures. Thus, the overall running time is $O(n+\widetilde{m})$,
where the additive $O(n)$ term accounts for the initialization of the vertex set, and in the case when $p\ll 1$, this matches $O(n+m)$ in expectation. 

For \ER graphs, this construction should be viewed mainly as a conceptual prototype rather than as a tout court replacement for the best specialized generators. In fact, exact generators with running time $O(n+m)$ are already known; the standard reference is the edge-skipping method via geometric jumps from Batagelj and Brandes \cite{batagelj2005efficient}, which generates $G_{n,p}$ directly as a simple graph and is used, for instance, in the implementation of \texttt{fast\_gnp\_random\_graph()} in NetworkX. By contrast, our proposed algorithm proceeds via an underlying multigraph plus an erasure projection and it is not intended to improve the asymptotic complexity of homogeneous \ER sampling itself.

Nonetheless, the study of the \ER case highlights the two main ingredients of Algorithm \ref{alg:er-simple}: (i) a temporal edge-arrival process that recovers the target random graph law, and (ii) an output-sensitive generation strategy based on simulating only realized edge events.
The remainder of the paper generalizes this scheme from uniform rates to heterogeneous rates.

\section{Rank-1 inhomogeneous random graphs}
\label{sec:model}

The family of \emph{rank-1} (or \emph{factorized}) inhomogeneous random graphs refers to models in which edges appear independently from each other and the connection probability $p_{ij}$ between two vertices $i$ and $j$ is a function of the product of their weights. This family is the natural baseline when the goal is to realize a prescribed expected degree sequence in an edge-independent random graph. Three standard parameterizations are especially common in the literature. If we write $L_n := \sum_{i=1}^n x_i$, their connection probabilities are
\[
p_{ij}^{\mathrm{NR}}
=
1-\exp\!\left(-\frac{x_i x_j}{L_n}\right),
\qquad
p_{ij}^{\mathrm{CL}}
=
\min\!\left\{\frac{x_i x_j}{L_n},1\right\}, \qquad
p_{ij}^{\mathrm{GRG}}
=
\frac{x_i x_j}{L_n+x_i x_j},
\qquad i\neq j.
\]
Here, NR denotes the Norros--Reittu simple projection \cite{norros2006conditionally}, CL the Chung--Lu model \cite{chung2002connected}, and GRG the generalized random graph \cite{britton2006generating}. Observe that $1-e^{-t}=t+O(t^2)$ and $\frac{t}{1+t}=t+O(t^2)$ for small $t$, so whenever $x_i x_j/L_n$ is small in all CL, NR, and GRG models
\begin{equation}\label{eq:first_order_approx}
    p_{ij}
    =
    \frac{x_i x_j}{L_n}
    +
    O\!\left(\frac{x_i^2x_j^2}{L_n^2}\right).    
\end{equation}

In fact, these models are asymptotically equivalent, as the size $n$ grows large, under certain conditions that can be found in \cite{Hofstad_2016} (Conditions 6.4(a)-(c)). In turn, NR, CL, and GRG induce the same target degrees to first order, matching in expectation the given sequence $x$, and they exhibit the same macroscopic limit behavior under standard regularity assumptions.

This asymptotic equivalence further motivates the relevance of the class of rank-1 models. The CL model is analytically tractable and computationally convenient, whereas the GRG encodes the target degree sequence in the least-biased manner. Instead, the Poissonized NR construction naturally lends itself to a temporal interpretation, and therefore, it is used in our current work. In this section, we provide a formal description of the model together with basic properties of its edge count and expected degrees.

\subsection*{Norros--Reittu construction}
\label{subsec:nr-model}

Let $V=[n]:=\{1,\dots,n\}$ and let $x=(x_i)_{i=1}^n\in\mathbb{R}_+^n$ be a prescribed weight sequence.
We adopt the Norros--Reittu construction \cite{norros2006conditionally} as our reference rank-1 model. The primary object is a Poissonian multigraph, from which a simple graph is obtained by projection.

\begin{definition}[Norros--Reittu model]
\label{def:nr-model}
Given $x \in \mathbb R_+^n$, the \emph{Norros--Reittu multigraph} of $x$ is the random multigraph $\widetilde G_n(x)=(V,\widetilde E)$ defined as follows. For each unordered pair $\{i,j\}$ let
\begin{equation}
    \widetilde E_{ij}\sim \Poi\!\left(\frac{x_i x_j}{L_n}\right), \quad \text{if } i\neq j, \qquad \text{and} \qquad \widetilde E_{ij}\sim \Poi\!\left(\frac{x_i^2}{2L_n}\right), \quad \text{if } i=j,
\end{equation}
independently for all $i \leq j$, and define $\widetilde E = \bigcup_{i \leq j} \widetilde E_{ij}$. Each realization of $\widetilde E_{ij}$ for $i<j$ contributes that many parallel edges between $i$ and $j$, whereas $\widetilde E_{ii}$ counts self-loops at vertex $i$.

The \emph{Norros--Reittu simple graph} of $x$ is the graph $G_n(x)=(V,E)$
obtained by deleting from $\widetilde G_n(x)$ all self-loops and replacing each positive multiplicity between distinct vertices by a single edge, namely, for all $i < j$
\begin{equation}
    \{i,j\}\in E
    \quad\Longleftrightarrow\quad
    \widetilde E_{ij} \ge 1.
\end{equation}
Equivalently, the edges of $G_n(x)$ are independent and satisfy
\[
\mathbb P(\{i,j\}\in E)
=
1-\exp\!\left(-\frac{x_i x_j}{L_n}\right),
\qquad \text{for all }i<j.
\]
\end{definition}

We write $\widetilde m:=|\widetilde E|$ for the total number of edges in the multigraph, counted with multiplicity and including loops, and $m:=|E|$ for the number of edges in the simple projection. For the corresponding degrees, we define $\widetilde D_i:= \sum_{j\neq i}\widetilde E_{ij}+2\widetilde E_{ii}$ and $D_i:=\sum_{j\neq i} 1_{\{\widetilde E_{ij}\ge 1\}}$. 

Typically, some conditions on the weight sequence $x$ are required for rank-1 models. Specifically, these conditions involve control over the first and second moments of the sequence, as in Condition 6.4 in \cite{Hofstad_2016}. However, the basic properties of the NR model used in this work requires weaker assumptions, which are directly implied by the above-mentioned stronger conditions in the literature. To avoid repeating them in each result, we introduce the following terminology.

\begin{definition}[Regularity conditions]\label{def:assumptions}
Consider a sequence $x \in \mathbb{R}_+^n$. We say that:
\begin{enumerate}
    \item $x$ is \textit{non-degenerate}, if
    \begin{equation}
        L_n = \sum_{j=1}^n x_j \to \infty 
        \qquad\text{as } n\to\infty.
    \end{equation}
    \item $x$ is \textit{hub-controlled}, if
    \begin{equation}
        \max_{i \in V} \frac{x_i}{L_n^{1/2}} \to 0 
        \qquad\text{as } n\to\infty.
    \end{equation}
\end{enumerate}
\end{definition}

The non-degeneracy places the model in the asymptotic setting of interest, ruling out the possibility that the total weight mass converges to a finite value. Whereas, the hub-control ensures that edge probabilities are uniformly small, allowing for a first-order approximation of the kind \eqref{eq:first_order_approx}. 

These assumptions are mild in many standard expected-degree scenarios. They are automatically satisfied, for instance, when weights are non-vanishing and uniformly bounded, or, more generally, whenever the empirical weight distribution has sufficiently light tails and no small set of vertices carries a macroscopic fraction of the total mass.\\

The Norros--Reittu construction yields exact formulas for the multigraph edge count and degree variables.
For the simple projection, analogous quantities differ only through the removal of loops and the collapse of multiple edges. The proposition below collects basic properties of the NR multigraph and simple graph that follow directly from the model definition. A detailed proof of these facts can be found in Appendix \ref{app:NR_edge_degree}.

\begin{proposition}\label{prop:NR_edge_degree}
    Let $x \in \mathbb{R}_{+}^n$. The number of edges $\widetilde m$ and the degree $\widetilde D_i$, for each $i \in V$, in the NR multigraph $\widetilde G_n(x)$ satisfy
    \begin{equation} \label{eq:result_multigraph}
            \widetilde m \sim \Poi\left(\frac{L_n}{2}\right) \qquad \text{and} \qquad \E[\widetilde D_i] = x_i.
    \end{equation}
    Moreover, assume that $x$ satisfies the regularity conditions in Definition \ref{def:assumptions}. Then, the number of edges $m$ and the degree $D_i$, for each $i \in V$, in the NR simple graph $G_n(x)$ satisfy
    \begin{equation} \label{eq:result_simplegraph}
        m = \widetilde m (1 + o_p(1)) \qquad \text{and} \qquad \E[D_i] = x_i(1+o(1)).
    \end{equation}
\end{proposition} 

\begin{proof}[Sketch of the proof]
    The results in \eqref{eq:result_multigraph} on the NR multigraph immediately follow from the model definition and edge independence. 
    
    Instead, the results in \eqref{eq:result_simplegraph} on the NR simple graph employ the regularity conditions. First, the amount of excess edges $\Delta_m:=\widetilde m - m$ that is lost from self-loops and duplicate edges erasure is proven to be negligible compared to the total amount of edges $\widetilde m$, via Markov's inequality. Second, using a first-order approximation, the expected degree of vertex $i$ in the simple graph approaches $x_i$ asymptotically.
\end{proof}

\subsection*{Generalizations}
\label{subsec:equiv-models}

Given a weight sequence $x$, more general inhomogeneous kernels may also be considered, properly adapting the generative algorithm described in the next section. For instance, models with
\[
p_{ij}=f(x_i,x_j)
\qquad\text{or}\qquad
p_{ij}=\frac{\kappa(x_i,x_j)}{n},
\]
for suitable functions $f$ or kernels $\kappa$.
Such generalizations are often useful when additional geometric, latent-space, or type-dependent effects are of interest.
However, in general, they no longer preserve the interpretation of the input sequence $x$ as the prescribed expected degree sequence, unless they are explicitly calibrated to do so.
For this reason, rank-1 models remain the most natural benchmark when the primary objective is expected-degree control.

\section{Event-driven algorithm for the NR model}
\label{sec:algorithm}

A standard fast baseline for sampling expected-degree random graphs is the \textit{edge-skipping} generator of Miller and Hagberg~\cite{miller2011efficient}, adapted from the fast algorithm generator introduced in \cite{batagelj2005efficient}. Their method targets the Chung-Lu model, but recently it has also been extended to generate Generalized random graphs \cite{li2025fast}. The intuitive idea is to exploit the sparsity of the graph by scanning candidate pairs in a suitable order and using geometric jumps to skip long runs of absent edges. Once the weights are sorted in decreasing order, the sampling phase runs in expected time $O(n+m)$. However, this requires a preprocessing step that sorts the input sequence, which costs $O(n\log n)$ time in the comparison model.

Here, instead, we revisit the static model in Section~\ref{sec:model} from a temporal \textit{edge-arrival} perspective, yielding an exact generative algorithm. The key observation is that the Norros--Reittu multigraph admits an equivalent event-driven representation: instead of sampling independently the multiplicity of every possible edge, one may first sample the total number of edge events and then generate each event by drawing its two endpoints independently from a distribution induced by the target weights.

More precisely, the construction can be viewed as a Poisson splitting mechanism. A global Poisson budget generates edge events, and each event is assigned independently to an unordered pair of endpoints according to the vertex distribution. This immediately yields a linear-time generation procedure in the number of realized events, after a linear preprocessing of the weights, and crucially avoids any sorting step on the input sequence.

\subsection{High-level idea}
\label{subsec:algo-idea}

Let $x \in \mathbb R_+^n$ and recall that $L_n = \sum_{i=1}^n x_i$. This sequence induces a natural probability distribution $\pi$ on the vertex set, which is obtained by normalizing $x$ by the total mass $L_n$. Namely,
\begin{equation}
        \pi(i):=\frac{x_i}{L_n}, \qquad i\in V.
\end{equation}
The algorithmic idea is the following:
\begin{enumerate}
    \item Sample the total number of edge events $K \sim \Poi\left(\frac{L_n}{2}\right).$
    \item For each event $r=1,\dots,K$, sample two endpoints $u,v \sim \pi$ independently.
    \item Insert the unordered pair $\{u,v\}$ into a multigraph. If $u=v$, this creates a self-loop; if the same pair is sampled multiple times, this creates parallel edges.
    \item If one wants a simple graph, delete self-loops and merge repeated copies of the same edge.
\end{enumerate}
The choice of the Poisson budget with parameter $L_n/2$ is canonical. Indeed, considering the output $\{u,v\}$ of a single edge event, the probability that the unordered set $\{u, v\}$ matches $\{i,j\}$, for any given $i,j \in V$ is
\begin{equation}
    \Prob(\{u, v\} = \{i,j\}) = 
    \begin{cases}
        2 \pi(i)\pi(j) = \frac{2 x_i x_j}{L_n^2} &\text{if $i \neq j$}\\
        \pi(i)^2 = \frac{x_i^2}{L_n^2} &\text{if $i = j$}        
    \end{cases}
\end{equation}
Hence, by Poisson thinning, the number of generated copies of the unordered pair $\{i,j\}$ is Poisson with mean $\frac{x_ix_j}{L_n}$, and the number of loops at $i$ is Poisson with mean $\frac{x_i^2}{2L_n}$. These are exactly the Norros--Reittu edge multiplicities from Definition~\ref{def:nr-model}. Therefore, the procedure above is not a heuristic approximation, but an exact alternative representation of the NR multigraph.

\subsection{Preprocessing of the weight sequence via the alias method}
\label{subsec:alias-preprocessing}

The only nontrivial ingredient in the event-driven representation is the repeated sampling of vertices from the discrete distribution $\pi$. Since this sampling must be performed independently for both endpoints of each edge event, it is crucial to support it in constant time per draw, after a one-time preprocessing step.

A standard solution is the \emph{alias method}, originally introduced by Walker and later refined by Vose into a linear-time preprocessing procedure \cite{walker1977efficient,vose1991linear}. The method represents the distribution $\pi$ by two arrays:
\begin{itemize}
    \item a cutoff array $q=(q_i)_{i=1}^n$ with $0\le q_i\le 1$,
    \item an alias array $A=(A_i)_{i=1}^n$ with $A_i\in[n]$.
\end{itemize}
Conceptually, the total mass of $x$ is redistributed into $n$ buckets, each with target mass $\mu:=L_n/n$. Bucket $i$ is assigned primarily to item $i$, but if $x_i<\mu$, then its missing mass is filled using the surplus of another item $A_i$, called its alias. Thus, within bucket $i$, a fraction $q_i$ corresponds to output $i$, while the remaining fraction $1-q_i$ corresponds to output $A_i$. Once the arrays $q$ and $A$ are available, one sample from $\pi$ is generated as follows:
\begin{enumerate}
    \item choose an index $I$ uniformly from $[n]$,
    \item sample $U\sim\mathrm{Unif}(0,1)$; if $U\le q_I$, return $I$, otherwise return its alias $A_I$.
\end{enumerate}
Therefore, each draw requires only a constant number of primitive operations: one uniform integer sample, one uniform real sample, two array accesses, and one comparison. 

\begin{figure}
    \centering
    \includegraphics[width=1\linewidth]{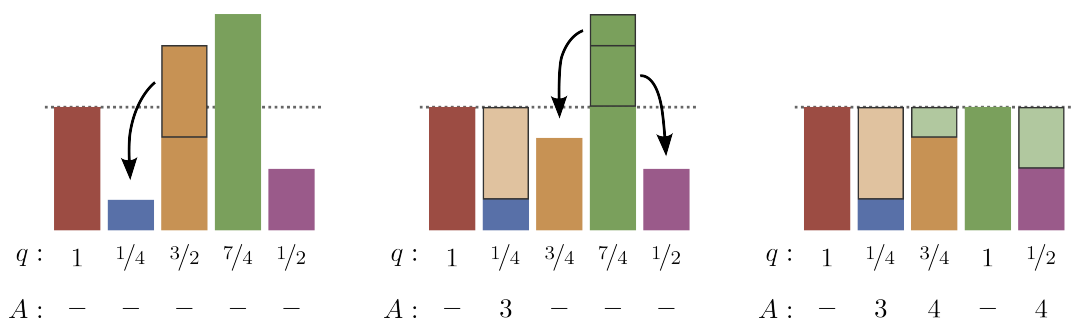}
    \caption{Illustrative example of the alias method for weights $x=\{4,1,6,7,2\}$. The average weight is $\mu=4$. The scaled weights are initialized as $q_i=x_i/\mu$, and the alias array $A$ is initially empty. In the first step, the deficit of item 2 is filled using surplus mass from item 3; in turn, item 3 becomes the alias of item 2, $A_2=3$, and $q_3$ is updated. The procedure is then iterated over remaining items with $q_i<1$ and no alias assigned yet, linking them to items with $q_j>1$, until all buckets reach level $1$ in scaled units (equivalently, height $\mu$ in the figure).}
\label{fig:alias_method}
\end{figure}

Figure \ref{fig:alias_method} shows how the mass redistribution of the alias method works in practice. Formally, we start by defining
\[
q_i := \frac{x_i}{\mu}.
\]
Items with $q_i<1$ are called \emph{small}, while those with $q_i>1$ are \emph{large}. We initialize two lists containing these two classes. At each step, one small item $s$ is paired with one large item $\ell$, and the bucket of $s$ is completed using surplus mass from $\ell$. Then, $\ell$ becomes the alias of $s$ and the residual mass of $\ell$ is updated:
\[
A_s=\ell, \qquad q_\ell \gets q_\ell - (1- q_s).
\]
After the update, $s$ is removed from the list of small items, whereas $\ell$ is returned to the appropriate list depending on whether its scaled residual mass is still above or below $1$. The process terminates when all buckets have been filled.
The correctness of the construction is immediate from the bucket decomposition. Indeed, for every $j\in[n]$,
\[
\Prob(\text{output } j)
= \frac{1}{n}q_j + \sum_{i:\,A_i=j}\frac{1}{n}(1-q_i)
= \frac{x_j}{n \mu} = \pi(j).
\]
Hence, the alias table $(q,A)$ provides an exact sampler for the target endpoint distribution.

The preprocessing requires $O(n)$ time and $O(n)$ memory. Indeed, after the initial $O(n)$-time scan that computes the scaled masses and partitions the items into lists of small and large entries, the construction proceeds by repeatedly pairing a small item with a large item. Each such pairing permanently finalizes the bucket of the small item and updates only the residual mass of the chosen large item. Hence, every iteration performs only constant work, and every index can enter and leave the working lists only a constant number of times. It follows that the total number of list operations is linear in $n$, yielding an $O(n)$ preprocessing algorithm and thereafter $O(1)$ time per sample. 

\subsection{Formal algorithm}
\label{subsec:formal-algorithm}

At this point, we can state the NR generator in Algorithm \ref{alg:nr-simple}, whose output exactly matches a random realization of the NR model. When the target object is the NR simple graph, one needs to delete loops and collapse edge multiplicities. In particular, since the simple graph is a deterministic projection of the multigraph, this implementation does not alter the output law. The projection from a multigraph to a simple graph can be performed online during generation by storing each non-loop edge in canonical order $(\min\{u,v\},\max\{u,v\})$ in a hash set. Figure \ref{fig:examples} shows three examples of the Algorithm in action.

\begin{algorithm}
\caption{Event-driven algorithm for the Norros--Reittu simple graph}
\label{alg:nr-simple}
\begin{algorithmic}[1]
\Require weight sequence $x_1,\dots,x_n \ge 0$
\Ensure simple graph $G_n(x)$ on vertex set $V=[n]$

\State Compute $L_n \gets \sum_{i=1}^n x_i$
\State Build an alias table for $\pi(i)=x_i/L_n$
\State Sample $K \sim \Poi(L_n/2)$
\State Initialize an empty set of edges $E$

\For{$r=1$ to $K$}
    \State Sample $u \sim \pi$
    \State Sample $v \sim \pi$
    \If{$u \neq v$}
        \State Insert $(\min\{u,v\},\max\{u,v\})$ into $E$
    \EndIf
\EndFor

\State \Return $(V,E)$
\end{algorithmic}
\end{algorithm}

\begin{figure}[t]
    \centering
    \includegraphics[width=0.95\linewidth]{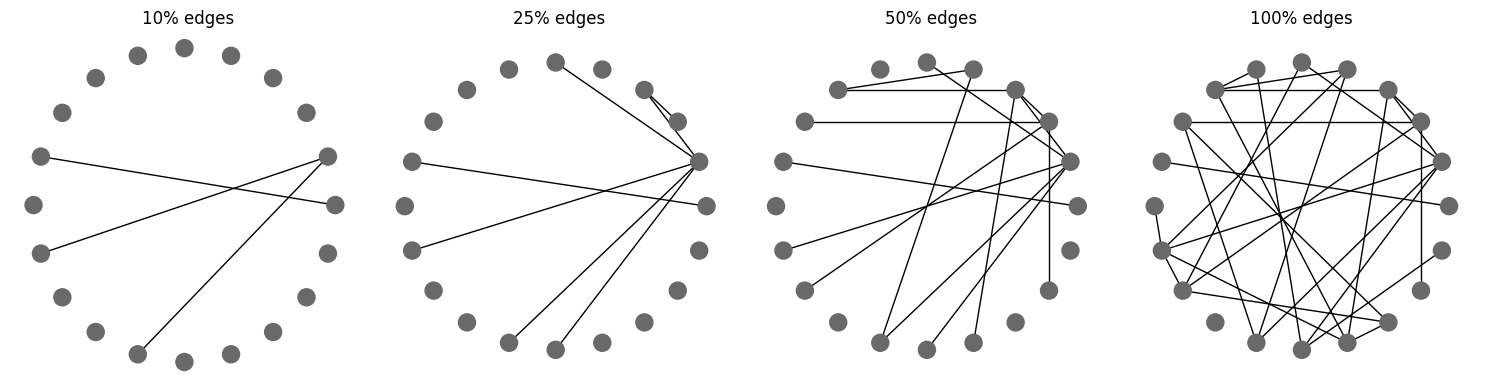}

    \vspace{0.8em}

    \includegraphics[width=0.95\linewidth]{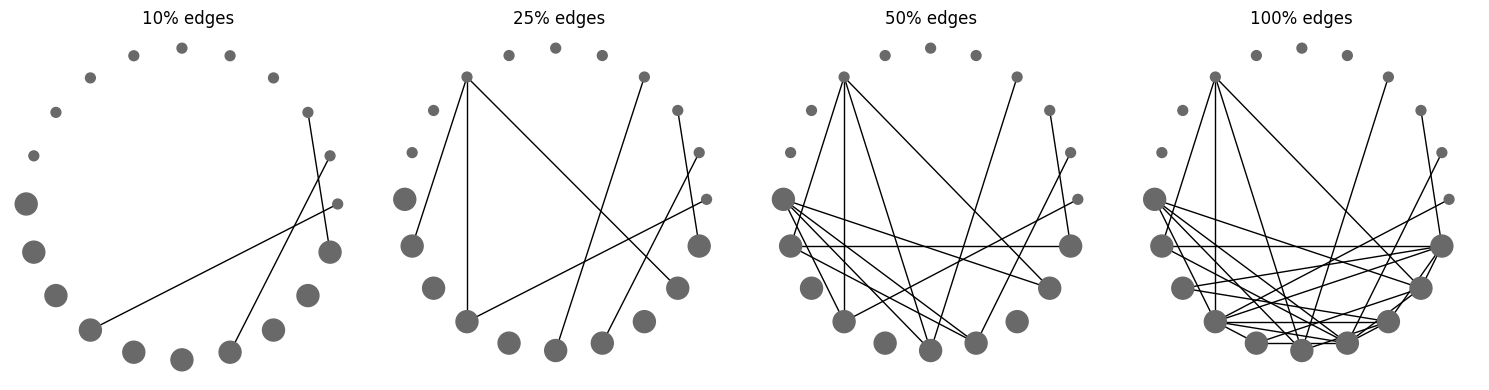}

    \vspace{0.8em}

    \includegraphics[width=0.95\linewidth]{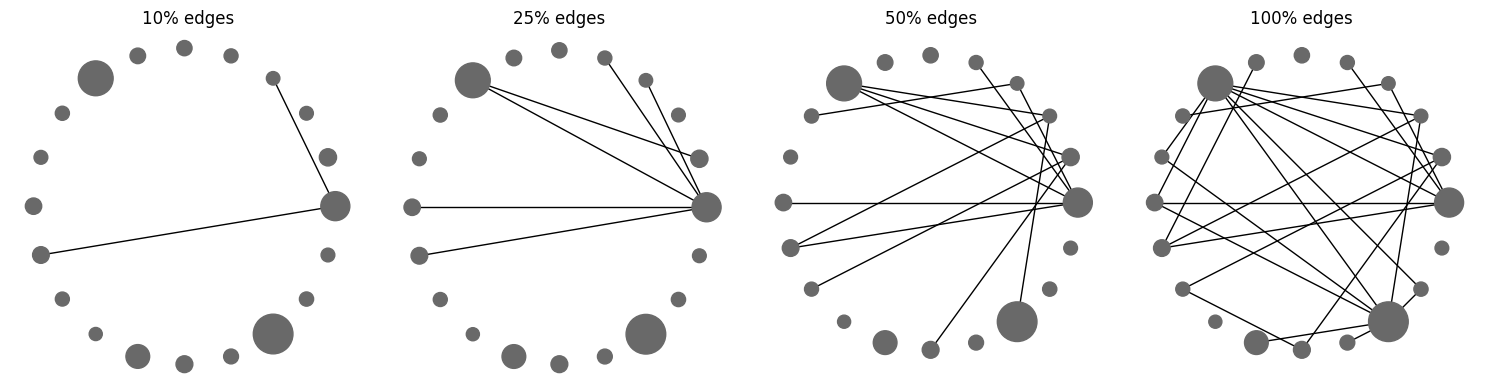}

    \caption{Temporal evolution of the event-driven generator at 10\%, 25\%, 50\%, and 100\% of the generated edge events. The vertex sizes are proportional to their weights. In all three examples, $n=20$ and the target average degree is $\mu=3$. Top: homogeneous weights, yielding an \ER graph. Middle: a two-block weight profile. Bottom: Pareto-distributed weights.}
    \label{fig:examples}
\end{figure}

The next theorem formally establishes that the event-driven procedure samples exactly the model introduced in Section~\ref{sec:model}.

\begin{theorem}[Exact Norros--Reittu graph law]
\label{thm:exact-multigraph-law}
Given any $x \in \mathbb{R}_+^n$, the random output of Algorithm~\ref{alg:nr-simple} is distributed exactly as the Norros--Reittu simple graph $G_n(x)$ from Definition~\ref{def:nr-model}.
\end{theorem}

\begin{proof}
For each fixed unordered pair $\{i,j\}$ with $i\neq j$, a single edge event in Algorithm~\ref{alg:nr-simple} selects that pair if and only if $(u,v)=(i,j)$ or $(j,i)$, which occurs with probability
\[
2\pi(i)\pi(j)=\frac{2x_i x_j}{L_n^2}.
\]
Since the total number of edge events is $K\sim\Poi(L_n/2)$, Poisson thinning implies that the total number of generated copies of $\{i,j\}$ is Poisson with mean
\[
\frac{L_n}{2}\cdot \frac{2x_i x_j}{L_n^2}=\frac{x_i x_j}{L_n}.
\]
Similarly, for a loop at vertex $i$, a single event produces it with probability $\pi(i)^2=x_i^2/L_n^2$, so the number of loops at $i$ is Poisson with mean
\[
\frac{L_n}{2}\cdot \frac{x_i^2}{L_n^2}=\frac{x_i^2}{2L_n}.
\]
Moreover, the counts associated with distinct unordered pairs are jointly independent by the splitting property of Poisson processes. Therefore, storing all sampled pairs, possibly including self-loops and duplicate edges, yields a multigraph with the same law as the NR multigraph from Definition~\ref{def:nr-model}.

Algorithm~\ref{alg:nr-simple} then applies the deterministic projection inside a hash set that avoids loops and merges repeated copies of the same edge. Therefore, its output is distributed exactly as the NR simple graph.
\end{proof}

Thus, the algorithmic procedure yields a generator for graphs with a prescribed expected degree sequence, provided the expected degree sequence satisfies the standard conditions detailed earlier.

\begin{corollary}[Expected-degree graph algorithm]
\label{cor:exact-simple-law}
Assume $x \in \mathbb{R}_+^n$ satisfies the regularity conditions in Definition \ref{def:assumptions}. Then, Algorithm \ref{alg:nr-simple} generates a simple random graph with asymptotic expected degrees $x$. 
\end{corollary}

\begin{proof}
    This is immediate from combining Theorem~\ref{thm:exact-multigraph-law} with Proposition~\ref{prop:NR_edge_degree}.
\end{proof}

\subsection{Running time and memory usage}
\label{subsec:complexity}

We next analyze the cost of the event-driven generator.

\begin{theorem}
\label{thm:running-time}
Under regularity conditions of Definition \ref{def:assumptions}, Algorithm \ref{alg:nr-simple} runs in time $O(n+m)$ with high probability, and it uses $O(n+m)$ space.
\end{theorem}

\begin{proof}
Computing $L_n$ and building the alias table both take $O(n)$ time and $O(n)$ memory.
Sampling $\widetilde m\sim\Poi(L_n/2)$ takes constant time.
Conditional on $\widetilde m$, each event requires two endpoint samples from the alias structure, hence $O(1)$ time, plus $O(1)$ expected update time in the chosen data structure.
Therefore, the total running time is
\[
O(n+\widetilde m)
\]
and, from Proposition \ref{prop:NR_edge_degree}, $\widetilde{m}$ matches $m$ with high probability.
For the simple graph implementation, storing edges in a hash set gives expected constant-time insertion and lookup per event, yielding an expected total time of $O(n+\widetilde m)$. The memory bound is immediate from the output representation.
\end{proof}

Theorem~\ref{thm:running-time} shows that the generation cost is linear in the input size plus the expected event budget. This reveals an interesting conclusion for the generation of a sparse random graph, when $m = \Theta(n)$.

\begin{corollary}[Fast generation of sparse expected-degree graphs.]\label{cor:fast_generation}
Suppose that $x$ satisfies regularity conditions of Definition \ref{def:assumptions} and the total weight mass is linear in the size of the graph, $L_n = \Theta(n)$. Then, Algorithm \ref{alg:nr-simple} generates a random graph with expected-degree sequence $x$ in time $O(n)$ with high probability.
\end{corollary}
\begin{proof}
The number of edges $m$ matches $\widetilde m$ asymptotically and $\widetilde m \sim \Poi(L_n/2)$. Therefore, if the total weight mass $L_n$ is of order $n$, then $m = \Theta(n)$ with high probability. The result then follows from Theorem \ref{thm:running-time}.
\end{proof}

\section{Conclusion and future directions}
\label{sec:conclusion}

Although the Miller--Hagberg generator and the present algorithm are exact for slightly different (but equivalent) modeling formulations, both address the task of generating graphs with a target expected-degree sequence. The present approach, however, follows a different principle. Rather than traversing the upper-triangular matrix of candidate pairs and skipping non-edges, it samples only the edge events that actually occur.
This distinction has a concrete algorithmic consequence. After an $O(n)$ preprocessing step to build the alias sampler for $\pi$, each realized event is generated in $O(1)$ time, with no need to sort the weights or maintain a monotone scan over candidate pairs. Thus, the preprocessing is genuinely linear in the input size, and the generation phase is driven directly by the (random) event budget itself. Our result is particularly valuable for sparse graphs. In this case, the running time of the current state-of-the-art algorithm, provided by Miller \& Hagberg is $O(n \log n)$, whereas, from Corollary \ref{cor:fast_generation}, our algorithm improves it by a factor $\log n$.
Therefore, the main advantage of our method is dual. First, it offers a novel, simple graph-generative scheme that shifts from the \textit{edge-skipping} to the \textit{edge-arrival} paradigm. In fact, the mechanism behind our algorithm is readily understandable and particularly easy to implement. Second, it entirely avoids monotone weight-sorting requirements, offering an algorithm with linear preprocessing and generation cost proportional to the realized event budget for expected-degree random graphs.\\

Several directions remain open. First, the event-driven viewpoint provides a genuine temporal or dynamical extension of the model, where edge arrivals are retained with their timestamps rather than projected immediately into a static snapshot, and they persist for a fixed or random time. This would lead to a rank-1 inhomogeneous temporal graph model, naturally connected to the literature on dynamic \ER graphs and related edge-Markov or edge-activation constructions  \cite{becker2023giant, broutin2024increasing, clementi2013rumor}. 
Second, the same principle behind Algorithm \ref{alg:nr-simple} extends naturally to directed graphs: given two weight sequences $x^{\mathrm{out}}$ and $x^{\mathrm{in}}$, one may build two alias samplers for the corresponding marginals and generate ordered endpoint pairs accordingly, so that the expected out-degrees and in-degrees are controlled separately \cite{durak2013scalable,yan2016asymptotics,bianchi2025mixing}.
Third, it would be interesting to investigate higher-order variants. Hypergraph analogs of rank-1 expected-degree models have already been considered \cite{kaminski2019clustering,saracco2025entropy}, and the present construction suggests that a similar Poisson arrival mechanism may provide an efficient generator for expected-degree hypergraphs. 

Finally, the algorithm appears especially amenable to parallelization: once the Poisson event budget is sampled, the arrivals can be split across $P$ workers, generated independently, and merged only at the end through the projection step. This suggests a simple parallel implementation using the techniques in \cite{hubschle-Schneider2019parallel}, which may be particularly attractive at a very large scale.

\bibliography{references}

@article{norros2006conditionally,
  title={On a conditionally Poissonian graph process},
  author={Norros, Ilkka and Reittu, Hannu},
  journal={Advances in Applied Probability},
  volume={38},
  number={1},
  pages={59--75},
  year={2006},
  publisher={Cambridge University Press},
  doi={10.1239/aap/1143936140}
}

@article{britton2006generating,
  title={Generating simple random graphs with prescribed degree distribution},
  author={Britton, Tom and Deijfen, Maria and Martin-L{\"o}f, Anders},
  journal={Journal of statistical physics},
  volume={124},
  number={6},
  pages={1377--1397},
  year={2006},
  publisher={Springer},
  doi={10.1007/s10955-006-9168-x}
}

@inproceedings{miller2011efficient,
author="Miller, Joel C.
and Hagberg, Aric",
editor="Frieze, Alan
and Horn, Paul
and Pra{\l}at, Pawe{\l}",
title="Efficient Generation of Networks with Given Expected Degrees",
booktitle="Algorithms and Models for the Web Graph",
year="2011",
publisher="Springer Berlin Heidelberg",
address="Berlin, Heidelberg",
pages="115--126",
doi="10.1007/978-3-642-21286-4_10"}

@article{chung2002connected,
  title={Connected components in random graphs with given expected degree sequences},
  author={Chung, Fan and Lu, Linyuan},
  journal={Annals of combinatorics},
  volume={6},
  number={2},
  pages={125--145},
  year={2002},
  publisher={Springer},
  doi={10.1007/PL00012580}
}

@book{Hofstad_2016, 
    place={Cambridge}, 
    series={Cambridge Series in Statistical and Probabilistic Mathematics}, 
    title={Random Graphs and Complex Networks (Vol. 1)}, 
    publisher={Cambridge University Press}, 
    author={Hofstad, Remco van der}, 
    year={2016}, 
    collection={Cambridge Series in Statistical and Probabilistic Mathematics},
    doi={10.1017/9781316779422}
}

@article{bollobas1980probabilistic,
title = {A Probabilistic Proof of an Asymptotic Formula for the Number of Labelled Regular Graphs},
journal = {European Journal of Combinatorics},
volume = {1},
number = {4},
pages = {311-316},
year = {1980},
issn = {0195-6698},
doi = {https://doi.org/10.1016/S0195-6698(80)80030-8},
author = {Béla Bollobás},
}

@article{bender1978asymptotic,
title = {The asymptotic number of labeled graphs with given degree sequences},
journal = {Journal of Combinatorial Theory, Series A},
volume = {24},
number = {3},
pages = {296-307},
year = {1978},
issn = {0097-3165},
doi = {https://doi.org/10.1016/0097-3165(78)90059-6},
author = {Edward A Bender and E.Rodney Canfield}
}

@InProceedings{li2025fast,
author="Li, Xuanchi
and Wang, Xin
and Kojaku, Sadamori",
editor="Cherifi, Hocine
and Rocha, Luis M.
and Cherifi, Chantal
and Ertem, Melissa Zeynep",
title="Fast Unbiased Sampling of Networks with Given Expected Degrees and Strengths",
booktitle="Complex Networks {\&} Their Applications XIV",
year="2026",
publisher="Springer Nature Switzerland",
address="Cham",
pages="330--341",
doi="10.1007/978-3-032-16645-6_28"
}

@article{erdos1960evolution,
  title={On the evolution of random graphs},
  author={Erd\"os, Paul and R{\'e}nyi, Alfr{\'e}d},
  journal={Publ. Math. Inst. Hungar. Acad. Sci},
  volume={5},
  number={1},
  pages={17--61},
  year={1960}
}

@article{roberts2018exceptional,
author = {Matthew I. Roberts and Batı Şeng{\"u}l},
title = {{Exceptional times of the critical dynamical Erdős–Rényi graph}},
volume = {28},
journal = {The Annals of Applied Probability},
number = {4},
publisher = {Institute of Mathematical Statistics},
pages = {2275 -- 2308},
year = {2018},
doi = {10.1214/17-AAP1357},
URL = {https://doi.org/10.1214/17-AAP1357}
}

@book{bollobas2001random, place={Cambridge}, edition={2}, series={Cambridge Studies in Advanced Mathematics}, title={Random Graphs}, publisher={Cambridge University Press}, author={Bollobás, Béla}, year={2001}, collection={Cambridge Studies in Advanced Mathematics}, doi={10.1017/CBO9780511814068}}

@article{broutin2024increasing,
author = {Nicolas Broutin and Nina Kamčev and G{\'a}bor Lugosi},
title = {{Increasing paths in random temporal graphs}},
volume = {34},
journal = {The Annals of Applied Probability},
number = {6},
publisher = {Institute of Mathematical Statistics},
pages = {5498 -- 5521},
year = {2024},
doi = {10.1214/24-AAP2097},
URL = {https://doi.org/10.1214/24-AAP2097}
}

@inproceedings{becker2023giant,
  author =	{Becker, Ruben and Casteigts, Arnaud and Crescenzi, Pierluigi and Kodric, Bojana and Renken, Malte and Raskin, Michael and Zamaraev, Viktor},
  title =	{{Giant Components in Random Temporal Graphs}},
  booktitle =	{Approximation, Randomization, and Combinatorial Optimization. Algorithms and Techniques (APPROX/RANDOM 2023)},
  pages =	{29:1--29:17},
  series =	{Leibniz International Proceedings in Informatics (LIPIcs)},
  ISBN =	{978-3-95977-296-9},
  ISSN =	{1868-8969},
  year =	{2023},
  volume =	{275},
  editor =	{Megow, Nicole and Smith, Adam},
  publisher =	{Schloss Dagstuhl -- Leibniz-Zentrum f{\"u}r Informatik},
  address =	{Dagstuhl, Germany},
  doi =		{10.4230/LIPIcs.APPROX/RANDOM.2023.29},

}

@article{batagelj2005efficient,
  title = {Efficient generation of large random networks},
  author = {Batagelj, Vladimir and Brandes, Ulrik},
  journal = {Phys. Rev. E},
  volume = {71},
  issue = {3},
  pages = {036113},
  numpages = {5},
  year = {2005},
  month = {Mar},
  publisher = {American Physical Society},
  doi = {10.1103/PhysRevE.71.036113}
}

@article{vose1991linear,
  author={Vose, M.D.},
  journal={IEEE Transactions on Software Engineering}, 
  title={A linear algorithm for generating random numbers with a given distribution}, 
  year={1991},
  volume={17},
  number={9},
  pages={972-975},
  doi={10.1109/32.92917}}

@article{walker1977efficient,
author = {Walker, Alastair J.},
title = {An Efficient Method for Generating Discrete Random Variables with General Distributions},
year = {1977},
issue_date = {Sept. 1977},
publisher = {Association for Computing Machinery},
address = {New York, NY, USA},
volume = {3},
number = {3},
issn = {0098-3500},
url = {https://doi.org/10.1145/355744.355749},
doi = {10.1145/355744.355749},
journal = {ACM Trans. Math. Softw.},
month = sep,
pages = {253–256},
numpages = {4}
}

@article{drobyshevskiy2019random,
author = {Drobyshevskiy, Mikhail and Turdakov, Denis},
title = {Random Graph Modeling: A Survey of the Concepts},
year = {2019},
issue_date = {November 2020},
publisher = {Association for Computing Machinery},
address = {New York, NY, USA},
volume = {52},
number = {6},
issn = {0360-0300},
doi = {10.1145/3369782},
journal = {ACM Comput. Surv.},
month = dec,
articleno = {131},
numpages = {36}
}

@Inbook{meyer2022,
author="Meyer, Ulrich
and Penschuck, Manuel",
editor="Bast, Hannah
and Korzen, Claudius
and Meyer, Ulrich
and Penschuck, Manuel",
title="Generating Synthetic Graph Data from Random Network Models",
bookTitle="Algorithms for Big Data: DFG Priority Program 1736",
year="2022",
publisher="Springer Nature Switzerland",
address="Cham",
pages="21--38",
isbn="978-3-031-21534-6",
doi="10.1007/978-3-031-21534-6_2",
url="https://doi.org/10.1007/978-3-031-21534-6_2"
}

@InProceedings{fabius2026,
author="Fabius, Anthony
and Pandey, Ujwal
and Lin, Dong
and Kaul, Yash
and Slota, George M.",
editor="Cherifi, Hocine
and Rocha, Luis M.
and Cherifi, Chantal
and Ertem, Melissa Zeynep",
title="Scalable Benchmark Graph Generation for the Maximum Cardinality Matching and Distance-1 Minimum Coloring Problems",
booktitle="Complex Networks {\&} Their Applications XIV",
year="2026",
publisher="Springer Nature Switzerland",
address="Cham",
pages="318--329",
isbn="978-3-032-16645-6",
doi="10.1007/978-3-032-16645-6_27"
}

@article{kaminski2021, title={Artificial Benchmark for Community Detection (ABCD)—Fast random graph model with community structure}, volume={9}, DOI={10.1017/nws.2020.45}, number={2}, journal={Network Science}, author={Kamiński, Bogumił and Prałat, Paweł and Théberge, François}, year={2021}, pages={153–178}}

@InProceedings{schwartz2024,
author="Schwartz, Catherine
and Savkli, Cetin
and Galante, Amanda
and Czaja, Wojciech",
editor="Cherifi, Hocine
and Rocha, Luis M.
and Cherifi, Chantal
and Donduran, Murat",
title="Tailoring Benchmark Graphs to Real-World Networks for Improved Prediction of Community Detection Performance",
booktitle="Complex Networks {\&} Their Applications XII",
year="2024",
publisher="Springer Nature Switzerland",
address="Cham",
pages="108--120",
isbn="978-3-031-53499-7",
doi="10.1007/978-3-031-53499-7_9"
}

@Inbook{zweig2016,
author="Zweig, Katharina A.",
title="Random Graphs as Null Models",
bookTitle="Network Analysis Literacy: A Practical Approach to the Analysis of Networks",
year="2016",
publisher="Springer Vienna",
address="Vienna",
pages="183--214",
isbn="978-3-7091-0741-6",
doi="10.1007/978-3-7091-0741-6_7",
url="https://doi.org/10.1007/978-3-7091-0741-6_7"
}

@article{casiraghi2021configuration,
  title={Configuration models as an urn problem},
  author={Casiraghi, Giona and Nanumyan, Vahan},
  journal={Scientific reports},
  volume={11},
  number={1},
  pages={13416},
  year={2021},
  publisher={Nature Publishing Group UK London},
  doi={10.1038/s41598-021-92519-y}
}

@article{bartlett2012epidemic, title={EPIDEMIC DYNAMICS ON RANDOM AND SCALE-FREE NETWORKS}, volume={54}, DOI={10.1017/S1446181112000302}, number={1–2}, journal={The ANZIAM Journal}, author={Bartlett, J. and Plank, M. J.}, year={2012}, pages={3–22}}

@book{barrat2008, place={Cambridge}, title={Dynamical Processes on Complex Networks}, publisher={Cambridge University Press}, author={Barrat, Alain and Barthélemy, Marc and Vespignani, Alessandro}, year={2008}, doi={10.1017/CBO9780511791383}}

@inproceedings{roucairol2022refutation,
author="Roucairol, Milo
and Cazenave, Tristan",
editor="Zhang, Yong
and Miao, Dongjing
and M{\"o}hring, Rolf",
title="Refutation of Spectral Graph Theory Conjectures with Monte Carlo Search",
booktitle="Computing and Combinatorics",
year="2022",
publisher="Springer International Publishing",
address="Cham",
pages="162--176",
isbn="978-3-031-22105-7",
doi="10.1007/978-3-031-22105-7_15"
}

@book{newman2018networks,
    author = {Newman, Mark},
    title = {Networks},
    publisher = {Oxford University Press},
    year = {2018},
    month = {07},
    isbn = {9780198805090},
    doi = {10.1093/oso/9780198805090.001.0001}
}

@article{ghavasieh2024diversity,
  title={Diversity of information pathways drives sparsity in real-world networks},
  author={Ghavasieh, Arsham and De Domenico, Manlio},
  journal={Nature Physics},
  volume={20},
  number={3},
  pages={512--519},
  year={2024},
  publisher={Nature Publishing Group UK London},
  doi={10.1038/s41567-023-02330-x}
}

@article{broido2019scale,
  title={Scale-free networks are rare},
  author={Broido, Anna D and Clauset, Aaron},
  journal={Nature communications},
  volume={10},
  number={1},
  pages={1017},
  year={2019},
  publisher={Nature Publishing Group UK London},
  doi={10.1038/s41467-019-08746-5}
}

@article{bringmann2019geometric,
title = {Geometric inhomogeneous random graphs},
journal = {Theoretical Computer Science},
volume = {760},
pages = {35-54},
year = {2019},
issn = {0304-3975},
doi = {https://doi.org/10.1016/j.tcs.2018.08.014},
author = {Karl Bringmann and Ralph Keusch and Johannes Lengler},
}

@incollection{penschuck2022recent,
  author       = {Manuel Penschuck and
                  Ulrik Brandes and
                  Michael Hamann and
                  Sebastian Lamm and
                  Ulrich Meyer and
                  Ilya Safro and
                  Peter Sanders and
                  Christian Schulz},
  editor       = {David A. Bader},
  title        = {Recent Advances in Scalable Network Generation},
  booktitle    = {Massive Graph Analytics},
  pages        = {333--376},
  publisher    = {Chapman and Hall/CRC},
  year         = {2022},
  doi          = {10.1201/9781003033707-16},

  }

@inproceedings{alam2016efficient,
  author       = {Md. Maksudul Alam and
                  Maleq Khan and
                  Anil Vullikanti and
                  Madhav V. Marathe},
  editor       = {John West and
                  Cherri M. Pancake},
  title        = {An efficient and scalable algorithmic method for generating large-scale random graphs},
  booktitle    = {Proceedings of the International Conference for High Performance Computing,
                  Networking, Storage and Analysis, {SC} 2016, Salt Lake City, UT, USA,
                  November 13-18, 2016},
  pages        = {372--383},
  publisher    = {{IEEE} Computer Society},
  year         = {2016},
  doi          = {10.1109/SC.2016.31},
}

@article{moreno2018scalable,
  author       = {Sebasti{\'{a}}n Moreno and
                  Joseph J. Pfeiffer III and
                  Jennifer Neville},
  title        = {Scalable and exact sampling method for probabilistic generative graph
                  models},
  journal      = {Data Min. Knowl. Discov.},
  volume       = {32},
  number       = {6},
  pages        = {1561--1596},
  year         = {2018},
  doi          = {10.1007/S10618-018-0566-X},
  bibsource    = {dblp computer science bibliography, https://dblp.org}
}

@inproceedings{hubschle-Schneider2019parallel,
  author       = {Lorenz H{\"{u}}bschle{-}Schneider and
                  Peter Sanders},
  editor       = {Michael A. Bender and
                  Ola Svensson and
                  Grzegorz Herman},
  title        = {Parallel Weighted Random Sampling},
  booktitle    = {27th Annual European Symposium on Algorithms, {ESA} 2019, Munich/Garching,
                  Germany, September 9-11, 2019},
  series       = {LIPIcs},
  pages        = {59:1--59:24},
  publisher    = {Schloss Dagstuhl - Leibniz-Zentrum f{\"{u}}r Informatik},
  year         = {2019},
  doi          = {10.4230/LIPICS.ESA.2019.59},
}

@inproceedings{blasius2019efficiently,
  author =	{Bl\"{a}sius, Thomas and Friedrich, Tobias and Katzmann, Maximilian and Meyer, Ulrich and Penschuck, Manuel and Weyand, Christopher},
  title =	{{Efficiently Generating Geometric Inhomogeneous and Hyperbolic Random Graphs}},
  booktitle =	{27th Annual European Symposium on Algorithms (ESA 2019)},
  pages =	{21:1--21:14},
  series =	{Leibniz International Proceedings in Informatics (LIPIcs)},
  ISBN =	{978-3-95977-124-5},
  ISSN =	{1868-8969},
  year =	{2019},
  volume =	{144},
  editor =	{Bender, Michael A. and Svensson, Ola and Herman, Grzegorz},
  publisher =	{Schloss Dagstuhl -- Leibniz-Zentrum f{\"u}r Informatik},
  address =	{Dagstuhl, Germany},
  doi =		{10.4230/LIPIcs.ESA.2019.21},
}

@inproceedings{clementi2013rumor,
author="Clementi, Andrea
and Crescenzi, Pierluigi
and Doerr, Carola
and Fraigniaud, Pierre
and Isopi, Marco
and Panconesi, Alessandro
and Pasquale, Francesco
and Silvestri, Riccardo",
editor="Bodlaender, Hans L.
and Italiano, Giuseppe F.",
title="Rumor Spreading in Random Evolving Graphs",
booktitle="Algorithms -- ESA 2013",
year="2013",
publisher="Springer Berlin Heidelberg",
address="Berlin, Heidelberg",
pages="325--336",
doi="10.1007/978-3-642-40450-4_28"
}

@inproceedings{durak2013scalable,
  author={Durak, Nurcan and Kolda, Tamara G. and Pinar, Ali and Seshadhri, C.},
  booktitle={2013 IEEE 2nd Network Science Workshop (NSW)}, 
  title={A scalable null model for directed graphs matching all degree distributions: In, out, and reciprocal}, 
  year={2013},
  volume={},
  number={},
  pages={23-30},
  doi={10.1109/NSW.2013.6609190}}

@article{yan2016asymptotics,
  title={Asymptotics in directed exponential random graph models with an increasing bi-degree sequence},
  author={Yan, Ting and Leng, Chenlei and Zhu, Ji},
  journal={The Annals of Statistics},
  pages={31--57},
  year={2016},
  publisher={JSTOR},
  doi={10.1214/15-AOS1343}
}

@article{bianchi2025mixing,
author = {Bianchi, Alessandra and Passuello, Giacomo},
title = {Mixing Cutoff for Simple Random Walks on the Chung–Lu Digraph},
journal = {Random Structures \& Algorithms},
volume = {66},
number = {1},
pages = {e21277},
year = {2025},
doi = {https://doi.org/10.1002/rsa.21277},
}

@article{saracco2025entropy,
  title={Entropy-based models to randomise real-world hypergraphs},
  author={Saracco, Fabio and Petri, Giovanni and Lambiotte, Renaud and Squartini, Tiziano},
  journal={Communications Physics},
  volume={8},
  number={1},
  pages={284},
  year={2025},
  publisher={Nature Publishing Group UK London},
  doi={10.1038/s42005-025-02182-2}
}

@article{kaminski2019clustering,
  title={Clustering via hypergraph modularity},
  author={Kami{\'n}ski, Bogumi{\l} and Poulin, Val{\'e}rie and Pra{\l}at, Pawe{\l} and Szufel, Przemys{\l}aw and Th{\'e}berge, Fran{\c{c}}ois},
  journal={PloS one},
  volume={14},
  number={11},
  pages={e0224307},
  year={2019},
  publisher={Public Library of Science San Francisco, CA USA},
  doi={10.1371/journal.pone.0224307}
}

\newpage
\appendix
\section*{APPENDIX}
\section{Basic properties of NR model}
\label{app:NR_edge_degree}

We collect here all the technical results that lead to the proof of Proposition \ref{prop:NR_edge_degree}.

\begin{lemma}[Multigraph edge count]
The total number of edges in $\widetilde G_n(x)$ satisfies
\begin{equation}
    \widetilde m \sim \Poi(L_n/2).
\end{equation}
\end{lemma}

\begin{proof}
The claim follows from the additivity of independent Poisson variables, which yields $\widetilde m \sim \Poi\left(\sum_{i<j}\frac{x_i x_j}{L_n}+\sum_{i=1}^n \frac{x_i^2}{2L_n}\right)$. Moreover,
\[ \sum_{i<j}x_i x_j = \frac{1}{2}\left(L_n^2-\sum_{i=1}^n x_i^2\right), \]
hence
\[\sum_{i<j}\frac{x_i x_j}{L_n}+\sum_{i=1}^n \frac{x_i^2}{2L_n}=\frac{1}{2L_n}\left(L_n^2-\sum_{i=1}^n x_i^2\right)+\frac{1}{2L_n}\sum_{i=1}^n x_i^2=\frac{L_n}{2}.\]
\end{proof}

\begin{lemma}[Multigraph expected degrees]
\label{lem:nr-multidegrees}
For every $i\in V$, the degree of vertex $i$ in the NR multigraph satisfies
\begin{equation}
    \E[\widetilde D_i] = x_i.
\end{equation}
\end{lemma}

\begin{proof}
By definition,
\[\widetilde D_i=\sum_{j\neq i}\widetilde E_{ij}+2\widetilde E_{ii}.\]
Then, from the additivity of independent Poisson variables, the total incidence count at vertex $i$ is the sum of two Poisson random variables
\[
\widetilde D_i \sim \Poi\left(\sum_{j \neq i} \frac{x_i x_j}{L_n}\right) + 2 \cdot \Poi\left(\frac{x_i^2}{2L_n}\right)
\]
Hence $\mathbb{E}[\widetilde D_i] = x_i$.
\end{proof}

For the sake of notation, we define $\xi_{ij}$ independent Bernoulli random variables with parameter $1- \exp\left(-\frac{x_i x_j}{L_n}\right)$ for each $i,j \in V$.

\begin{lemma}[Simple graph edge count]
The total number of edges in $G_n$ satisfies
\begin{equation}
    m \sim \sum_{i<j}\xi_{ij}.
\end{equation}
In particular, $m$ is Poisson-binomial with mean $\mathbb E[m] = \sum_{i < j}\left(1-\exp\left(-\frac{x_i x_j}{L_n}\right)\right)$. 
\end{lemma}

\begin{proof}
By construction, the simple NR projection contains the edge $\{i,j\}$ if and only if $\widetilde E_{ij} \ge 1$. Therefore, for each $j\neq i$,
\[ 1_{\{\{i,j\}\in E\}} = 1_{\{\widetilde E_{ij}\ge 1\}} \sim \text{Bernoulli}\left(1-\exp\left(-\frac{x_i x_j}{L_n}\right)\right).\]
Since edge occurrences are independent, the indicators $\xi_{ij}$,
are also independent for each $i,j$, and
\[ m=\sum_{i < j} 1_{\{\{i,j\}\in E\}} \sim \sum_{i < j}\xi_{ij}. \]
\end{proof}

\begin{lemma}[Simple graph degrees]
\label{lem:nr-simple-degrees}
For every $i\in V$, the degree of vertex $i$ in the NR simple graph satisfies
\begin{equation}
    D_i \sim \sum_{j\neq i} \xi_{ij}.
\end{equation} 
In particular, $D_i$ is Poisson-binomial with mean $\mathbb E[D_i]=\sum_{j\neq i}\left(1-\exp\!\left(-\frac{x_i x_j}{L_n}\right)\right)$. Moreover, if $x$ satisfies the regularity conditions, then $\mathbb E[D_i]=x_i(1+o(1))$.
\end{lemma}

\begin{proof}
By construction, the simple NR projection contains the edge $\{i,j\}$ if and only if $\widetilde E_{ij} \ge 1$. Therefore, for each $j\neq i$,
\[ 1_{\{\{i,j\}\in E\}} = 1_{\{\widetilde E_{ij}\ge 1\}} \sim \text{Bernoulli}\left(1-\exp\left(-\frac{x_i x_j}{L_n}\right)\right).\]
Since the variables $(\widetilde E_{ij})_{j\neq i}$ are independent, the indicators $\xi_{ij}:=1_{\{\widetilde E_{ij}\ge 1\}}$,
are also independent for each $i,j$, and
\[ D_i=\sum_{j\neq i} 1_{\{\{i,j\}\in E\}}\sim \sum_{j\neq i}\xi_{ij}. \]

Next, recall that $0 \leq t - (1-e^{-t}) \leq \frac{t^2}{2}$ whenever $t \geq 0$; thus,
\[ 
\mathbb E[D_i] = \sum_{j\neq i}\frac{x_i x_j}{L_n} + O\left(\sum_{j\neq i}\frac{x_i^2x_j^2}{L_n^2}\right) 
\quad \Longrightarrow \quad
\mathbb E[D_i]=x_i-\frac{x_i^2}{L_n}+O\!\left(\frac{x_i^2}{L_n^2}\sum_{j=1}^n x_j^2\right).
\]
Finally, under the stated assumption, $\max_{i \in V} x_i = o(\sqrt{L_n})$. Then, $x_i^2/L_n = o(1)$ and 
\[
\frac{x_i^2}{L_n^2}\sum_{j=1}^n x_j^2 \leq \frac{x_i^2 \max_{j \in V} x_j}{L_n^2} \sum_{j=1}^n x_i \leq x_i \frac{(\max_{i \in V} x_i)^2}{L_n} = o(x_i).
\]
Hence, $\mathbb E[D] = x_i(1+o(1))$.
\end{proof}

The previous lemma shows that the target sequence $x$ is matched exactly in expectation in the multigraph model and, under mild regularity conditions, also asymptotically in simple graph projection.

We next compare the edge budgets of the multigraph and the simple graph. Let
\[
\Delta_m:=\widetilde m-m.
\]
The quantity $\Delta_m$ counts all edges lost when projecting from the multigraph to the simple graph, namely all self-loops together with the excess multiplicity of parallel edges. 

\begin{lemma}[Excess edges]
\label{lem:excess-edges}
If $x \in \mathbb{R}_+^n$ satisfies the regularity conditions, then the excess edge count is asymptotically negligible compared to the number of edges in the multigraph $\widetilde G_n$, that is,
\[
\Delta_m=o_{\mathbb P}(\widetilde m).
\]
\end{lemma}

\begin{proof}
By construction, $\widetilde m$ counts all edges in the multigraph, including loops and repeated edges, while $m$ counts only one edge for each unordered pair $i<j$ with positive multiplicity and ignores loops altogether. Hence
\[
\Delta_m
=
\sum_{i=1}^n \widetilde E_{ii}
+
\sum_{i<j}\Bigl(\widetilde E_{ij}-1_{\{\widetilde E_{ij}\ge 1\}}\Bigr).
\]
Taking expectations and using
$\E[\widetilde E_{ii}]=\frac{x_i^2}{2L_n},\E[\widetilde E_{ij}]=\frac{x_i x_j}{L_n},\Prob(\widetilde E_{ij}\ge 1)=1-\exp\!\left(-\frac{x_i x_j}{L_n}\right)$, gives
\[
\E[\Delta_m] = \frac{1}{2L_n}\sum_{i=1}^n x_i^2 + \sum_{i<j} \left(\frac{x_i x_j}{L_n} - \left(1-\exp\left(-\frac{x_i x_j}{L_n}\right)\right) \right).
\]

Thus, from the classical inequality $0\le t-(1-e^{-t})\le \frac{t^2}{2}$, taking $t=x_i x_j/L_n$ yields
\[
\E[\Delta_m] \le \frac{1}{2L_n}\sum_{i=1}^n x_i^2 + \frac{1}{2L_n^2}\sum_{i<j}x_i^2x_j^2 \le \frac{1}{2L_n}\sum_{i=1}^n x_i^2 + \frac{1}{4L_n^2}\left(\sum_{i=1}^n x_i^2\right)^2.
\]
In particular, 
\[\frac{\sum_{i=1}^n x_i^2}{L_n} \leq \max_{i\in V} x_i = o(L_n), \qquad \text{and} \qquad \frac{(\sum_{i=1}^n x_i^2)^2}{L_n^2} \leq (\max_{i\in V} x_i )^2 = o(L_n).\]
Then, since $\Delta_m\ge 0$, Markov's inequality gives, for every $\varepsilon>0$,
\[
\Prob(\Delta_m>\varepsilon L_n) \le \frac{\E[\Delta_m]}{\varepsilon L_n}\to 0.
\]
Thus $\Delta_m=o_{\mathbb P}(L_n)$, and, as $L_n \to \infty$, the conclusion follows from $\widetilde m\sim\Poi(L_n/2)$.
\end{proof}

\end{document}